\documentclass[10pt]{article}
 
\usepackage{fullpage}
\usepackage[all=normal,bibliography=tight]{savetrees}
\usepackage{microtype}
\usepackage{todonotes}
\usepackage{amsthm}
\usepackage{amssymb,bbm}
\usepackage{microtype}
\usepackage{xspace}
\usepackage{todonotes}
\usepackage{comment}
\usepackage{enumerate}
\usepackage{tikz}
\usepackage{graphicx}
\usepackage[absolute]{textpos}
\usetikzlibrary{arrows, decorations.markings}
\usetikzlibrary{fit}					
\usetikzlibrary{backgrounds}

\usepackage{amsmath,amssymb,amsthm}
\usepackage{graphicx}
\usepackage{caption}
\usepackage{subcaption}

\usepackage{pdflscape}
\usepackage{afterpage}

\usepackage{url}

\newcommand{\uselipics}{no}

\usepackage{multirow}

\newcommand{\iflipics}[2]{\ifthenelse{\equal{\uselipics}{yes}}{#1}{#2}}
\newcommand{\onlylipics}[1]{\iflipics{#1}{}}

\def\cqedsymbol{\ifmmode$\lrcorner$\else{\unskip\nobreak\hfil
\penalty50\hskip1em\null\nobreak\hfil$\lrcorner$
\parfillskip=0pt\finalhyphendemerits=0\endgraf}\fi}

\newcommand{\Oh}{\ensuremath{\mathcal{O}}}
\newcommand{\bag}{\beta}

\newcommand{\rotcat}[2]{\hline\multirow{#1}{*}{#2}}

\numberwithin{equation}{section}
\newtheorem{theorem}{Theorem}[section]
\newtheorem{lemma}[theorem]{Lemma}

\theoremstyle{definition}

\title{Finding Hamiltonian Cycle in Graphs of Bounded Treewidth:
  Experimental Evaluation%
\thanks{Supported by the ``Recent trends in kernelization: theory and experimental evaluation'' project, carried out within the Homing programme of the Foundation for Polish Science co-financed by the European Union under the European Regional Development Fund.}}

\author{Micha\l{} Ziobro\thanks{Theoretical Computer Science Department, Faculty of Mathematics and Computer Science, Jagiellonian University, Krak\'{o}w, Poland. \texttt{michal.18.ziobro@student.uj.edu.pl}}
  \and
    Marcin Pilipczuk\thanks{Institute of Informatics, University of Warsaw, Poland. 
      \texttt{malcin@mimuw.edu.pl}}}

\date{}

\begin{document}

\maketitle

\begin{abstract}
The notion of treewidth, introduced by Robertson and Seymour in their seminal Graph Minors series, turned out to have tremendous impact on graph algorithmics.
Many hard computational problems on graphs turn out to be efficiently solvable in graphs of bounded treewidth: graphs that can be sweeped with separators
of bounded size. These efficient algorithms usually follow the dynamic programming paradigm.

In the recent years, we have seen a rapid and quite unexpected development of involved techniques for solving various computational problems
in graphs of bounded treewidth. 
One of the most surprising directions is the development of algorithms for connectivity problems that have only single-exponential
dependency (i.e., $2^{\Oh(t)}$) on the treewidth in the running time bound, as opposed to slightly superexponential (i.e., $2^{\Oh(t \log t)}$)
  stemming from more naive approaches.
In this work, we perform a thorough experimental evaluation of these approaches in the context of one of the most classic connectivity problem,
  namely \textsc{Hamiltonian Cycle}. 
\end{abstract}

\section{Introduction}

The problem of finding {\sc Hamiltonian Cycle} in graph is one of the oldest and best known $\mathcal{NP}$-complete problems.
It was intensely studied together with its more generic optimization version {\sc Traveling Salesman Problem}. Early and important
result on this problem was the dynamic programming algorithm invented independently by Bellman~\cite{bellman1958combinatorial} and Held and Karp~\cite{held1962dynamic},
running in time $O(2^n n^2)$. The exponential factor of this running time bound remains the best known for deterministic
algorithms up to today, and a faster randomized Monte Carlo algorithm has been shown only recently by Bj{\"o}rklund \cite{bjorklund2010determinant}.
Faster algorithms were also obtained for some special cases, like graphs with bounded degree \cite{cygan2013fast,bjorklund2010trimmed} or claw-free graphs \cite{broersma2009fast}.
		
An important class of graphs in which many combinatorial problems can be solved more efficiently, are graphs of bounded \emph{treewidth}.
Treewidth, introduced by Robertson and Seymour in their Graph Minors project~\cite{robertson1984graph}, measures how the input graph resembles a tree,
 or how can be covered be a set of bounded-sized bags organized in tree like structure which we call \emph{tree decomposition}.
 It has proven to be very useful for dealing with $\mathcal{NP}$-hard problems;
 for example, given an $n$-vertex graph $G$ and its tree decomposition of width $t$, one can solve the \textsc{Maximum Independent Set} problem in $G$
in time $2^t \cdot t^{\Oh(1)} \cdot n$. We refer to~\cite{cygan2015parameterized} for more examples of algorithms on graphs of bounded treewidth.

Essentially every algorithm for graphs of bounded treewidth follows the paradigm of dynamic programming: it gradually (in a bottom-to-top fashion on the tree decomposition)
builds partial solutions in subgraphs of the input graph. Using the fact that a bag in a tree decomposition is a separator, in many combinatorial problems it suffices to keep
only a bounded (by a function of the width of the decomposition) number of partial solutions in each step of the algorithm.
To illustrate this concept, consider a separation $(A,B)$ in a graph $G$ with $S = A \cap B$ (i.e., $A, B \subseteq V(G)$ are two sets with $A \cup B = V(G)$ and no edge between $A \setminus B$
and $B \setminus A$), and think of a dynamic programming algorithm that processed already the graph $G[A]$, but has not yet touched $B \setminus A$.
Observe that a partial solution $X \subseteq A$ to the \textsc{Maximum Independent Set} problem interacts with $B \setminus A$ only via the set $S$.
Consequently, it suffices to store, for every $X_S \subseteq S$, an independent set $X \subseteq A$ of maximum possible size satisfying $X \cap S = X_S$.
If the separator $S$ is of size at most $t$, it leads to $2^t$ bound on the size of the memoization table in the dynamic programming algorithm.

In the \textsc{Hamiltonian Cycle} problem, the natural state space for the dynamic programming algorithm is a bit more complex. A partial solution in $G[A]$
would be a set of vertex-disjoint paths $\mathcal{P}$ that all have endpoints in $S$ and together visit every vertex of $A \setminus B$. 
To complete the partial solution $\mathcal{P}$ to a Hamiltonian cycle $H$ in $G$, it seems essential to remember not only which vertices
of $S$ are visited by $\mathcal{P}$ and which are the endpoints of paths in $\mathcal{P}$, but also how the paths of $\mathcal{P}$ pair up their endpoints in $S$
(see Figure~\ref{fig:sep}).
This last piece of information leads to $2^{\theta(t \log t)}$ states for separator $S$ of size $t$.

\begin{figure}[tb]
\includegraphics[width=\textwidth]{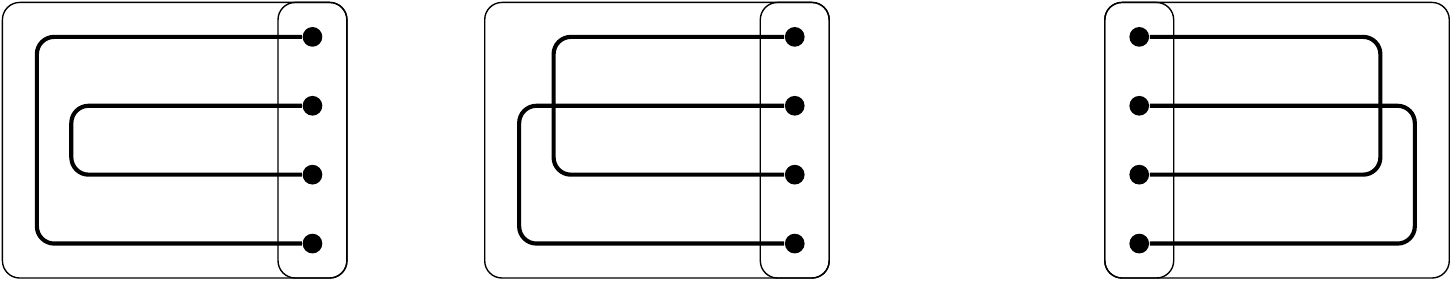}
\caption{A separator $S$ with two possible partial solutions on the left.
  Only the first one forms a Hamiltonian cycle with the partial solution
    on the right, despite that in all of them the vertices on the separator
    have degree $1$.}\label{fig:sep}
    \end{figure}

Up to late 2010, almost all known algorithms for combinatorial problems in graphs of bounded treewidth follow the naive approach outlined above,
and researchers' effort focused mostly on speeding up computations in the so-called \emph{join nodes} of the decomposition (see e.g.~\cite{rooij}).%
\footnote{A join node of a decomposition corresponds to a node of the underlying tree of the tree decomposition of degree at least $3$; intuitively, it corresponds
to a bounded-size separator that splits the graph into more than $2$ pieces, and in the dynamic programming algorithm one needs to merge information from at least
two of such pieces.}
In 2010, Lokshtanov, Marx, and Saurabh proved that many such algorithms have optimal dependency on treewidth~\cite{LMS} (under strong complexity assumptions)
and provided a framework for proving similar lower bounds for complexities of the type $2^{\theta(t \log t)}$~\cite{LMS2}. 
However, providing such a tight lower bound for the connectivity problems such as \textsc{Hamiltonian Cycle} in graphs of bounded treewidth remained elusive.

Quite unexpectedly, a year after it turned out that there is a reason for this lack of progress, and a Monte Carlo algorithm with running time
$4^t n^{\Oh(1)}$ for finding a Hamiltonian cycle in a graph with a given tree decomposition of width $t$ has been reported~\cite{cygan2011solving}.
The work~\cite{cygan2011solving} introduced a framework called Cut\&Count that provided randomized single-exponential (i.e., with running time bound
of the form $2^{\Oh(t)} n^{\Oh(1)}$) algorithms for many connectivity problems in graphs of bounded treewidth.
The key idea of the Cut\&Count method is to replace the original connectivity requirement with a different counting-mod-2 task, and ensure correctness
via the Isolation Lemma~\cite{IsolationLemma}.

In following years, a good understanding of the aforementioned improvement has been obtained by Bodlaender et al~\cite{bodlaender2012solving}.
In the language of \textsc{Hamiltonian Cycle}, a linear algebra argument shows that it suffices only to keep $4^t$ partial solutions instead 
of the naive bound of $2^{\Oh(t \log t)}$; if the memoization table grows too large, an algorithm based on Gaussian elimination is able to prune
provably unnecessary states. 
Cygan et al.~\cite{cygan2013fast} provided a better basis for the Gaussian elimination step and improved the bound for the number of states
for \textsc{Hamiltonian Cycle} to $(2+\sqrt{2})^t$. Furthermore, in~\cite{cygan2013fast} a matching lower bound is shown.
Due to the linear algebraic nature of the argument, this approach has been dubbed in the literature as the \emph{rank-based approach}.

In~\cite{cygan2011solving}, an involved fast convolution algorithm has been applied to obtain the $4^t n^{\Oh(1)}$ running time bound even
in computations at join nodes. 
The need to execute Gaussian elimination in~\cite{bodlaender2012solving} and treat join nodes in a more direct fashion in both algorithms
of~\cite{bodlaender2012solving,cygan2013fast} yield worse theoretical running time bounds.
Thus, the algorithm~\cite{cygan2011solving} remains theoretically fastest in graphs of bounded treewidth to this date.

Following a recent trend in multivariate algorithmics to experimentally evaluate parameterized algorithms
(led by a growing popularity of the Parameterized Algorithms and Computational Experiments Challenge \cite{dellpace,dell2017first}),
  in this work we thoroughly evaluate the aforementioned algorithms for \textsc{Hamiltonian Cycle}.
A direct inspiration for our work is the work of Fafianie et al~\cite{fafianie2015speeding} 
that provided an experimental comparison of the naive and rank-based approaches for \textsc{Steiner Tree} (i.e., without considering the Cut\&Count approach). 
That is, in this work we include Cut\&Count and we compare the following four approaches.
\begin{description}
\item[naive] The aforementioned naive approach with $2^{\Oh(t \log t)}$ bound on the number of states.
\item[rank-based] The approach of~\cite{bodlaender2012solving}, that is, the naive approach with pruning of the state space
leading to $4^t$ size bound.
\item[rank-based with improved basis] The approach of~\cite{cygan2013fast}, that is, the rank-based approach with the improved basis
yielding the size bound $(2+\sqrt{2})^t$.
\item[Cut\&Count] The Cut\&Count algorithm of~\cite{cygan2011solving}.
\end{description}
Furthermore, as a baseline, we have compared our algorithms to the \texttt{hamiltonian\_cycle()} method of the SageMath
package~\cite{sage}. 
Behind the scenes, the SageMath
implementation uses a Traveling Salesperson Problem solver relying on Gurobi as a linear programming solver.

As observed in~\cite{cygan2011solving}, the application of the Isolation Lemma in the Cut\&Count method yields a relatively high
polynomial factor in the running time bound, but one can replace its usage with computations over a field of characteristic $2$
and randomization via the Schwartz-Zippel lemma. This replacement leads again to linear dependency on the graph size in the running time bound.
We follow this path.
However, as has been overlooked in~\cite{cygan2011solving}, the fast convolution algorithm at join nodes in the $4^t n^{\Oh(1)}$-time algorithm
does not support computations over a field of characteristic $2$, as it requires division by $2$. 
Our theoretical contribution in this paper is a method around this obstacle, essentially showing that it is sufficient to perform the convolution
over the ring of polynomials $\mathbb{Z}[x]$. This is described in Section~\ref{ss:cut-and-count} and leads to the following conclusion.
\begin{theorem}\label{thm:cc}
There exists a Monte Carlo algorithm that, given an $n$-vertex graph $G$ together with its tree decomposition of width $t$,
      solves \textsc{Hamiltonian Cycle} on $G$ in time $4^t \cdot n \cdot (t \log n)^{\Oh(1)}$.
\end{theorem}

In Section~\ref{sec:impl} we discuss implementation details of the algorithms.
Section~\ref{sec:setup} discuss experiment setup and Section~\ref{sec:results} discuss results.
We conclude in Section~\ref{sec:conc}.

\section{Theory and implementation details}\label{sec:impl}

\subsection{Tree decompositions}

For more background on tree decompositions and dynamic programming algorithms using them, we refer to~\cite{cygan2015parameterized}.
Here, we recall only the basic notions.

For a graph $G$, a \emph{tree decomposition} is a pair $(T,\beta)$ where $T$ is a tree and $\beta$ assigns to every node $t \in V(T)$
a set $\beta(t) \subseteq V(G)$ called \emph{a bag} with the following invariants: (i) for every $v \in V(G)$, the set $\{t \in V(T) | v \in \beta(t)\}$
is nonempty and connected in $T$, (ii) for every $uv \in E(G)$ there exists $t \in V(T)$ with $u,v \in \beta(t)$.
The width of the tree decomposition is the maximum size of a bag, minus one, and the treewidth of a graph is the minimum possible width of its tree decomposition.

As in multiple previous algorithms, it is convenient to describe dynamic programming algorithms on a special type of decompositions, called \emph{nice}.
A \emph{nice tree decomposition} is a rooted tree decomposition for which the bag of the root is empty and every node is of one of the following types:
\begin{description}
	\item[Leaf node] is a node $t$ with no children and $\beta(t) = \emptyset$.
	\item[Introduce vertex node] is a node $t$ with unique child $t'$ and a vertex $v$ such that $\bag(t) = \bag(t') \uplus \{v\}$.
	\item[Forget vertex node] is a node $t$ with unique child $t'$ and a vertex $v \in \bag(t')$ such that $\bag(t) = \bag(t') \setminus \{v\}$.
	\item[Join node] is a node $t$ with exactly two children $t_1$ and $t_2$ with $\bag(t) = \bag(t_1) = \bag(t_2)$.
\end{description}
For a note $t \in V(T)$, we define $\gamma_{\downarrow}(t)$ to be the union of $\beta(t')$ over $t'$ being descendants of $t$ in $T$.
Furthermore, let $G_t$ be the graph $G[\gamma_\downarrow(t)]-E(G[\beta(t)])$ (i.e., 
    we exclude the edges inside the bag $\beta(t)$).

Additionally, in our case it is convenient to precede every \textbf{forget vertex node} with a sequence of \textbf{introduce edge nodes}. That is, 
for a \textbf{forget node} $t$ with child $t'$ and forgotten vertex $v$, we take $E_{t,v}$ to be the set of edges of $G$ that connect $v$ with
vertices of $\beta(t) \setminus \{v\}$, subdivide the edge $tt'$ in $E(T)$ $|E_{t,v}|$ times, labeling the new nodes $\{t_e | e \in E_{t,v}\}$,
and set $\beta(t_e) = \beta(t')$. The graphs $G_{t_e}$ are defined as follows: if $t''$ is the unique child of $t_e$, then
$G_{t_e} = G_{t''} \cup \{e\}$.

The intuition of this step is as follows: there is a significant difference between the graphs $G_{t'}$ and $G_t$,
 namely $E(G_t) = E(G_{t'}) \cup E_{t,v}$. We split this change into $|E_{t,v}|$ steps, adding edges one by one.

All our implementations start with preparing a nice tree decomposition with the \textbf{introduce edge nodes}.

\subsection{Naive approach}
Given a note $t$ in a nice tree decomposition $(T,\beta)$, a \emph{partial solution} is a family $\mathcal{P}$
of vertex-disjoint paths in $G_t$ such that (i) every vertex of $\gamma_\downarrow(t) \setminus \beta(t)$ is visited by some path in $\mathcal{P}$, and
(ii) every path in $\mathcal{P}$ has both its endpoints in $\beta(t)$.
For a partial solution $\mathcal{P}$ at note $t$, we define the following objects:
\begin{description}
\item[a bucket $b$] is a function $b : \beta(t) \to \{0,1,2\}$ that assigns to every vertex $v \in \beta(t)$ its degree in the union of $\mathcal{P}$;
\item[a pairing $E$] is a family of disjoint two-element subsets of
$b^{-1}(1)$ that pairs up the endpoints of the same path in $\mathcal{P}$.
\end{description}
The pair $(b,E)$ is the \emph{state} of $\mathcal{P}$.
The crucial observation is that among partial solutions with the same state, it suffices to memoize only one.
Note that for a given bucket $b$ with $\ell = |b^{-1}(1)|$, there are $(\ell-1)\cdot(\ell-3)\cdot 3 \cdot 1$ possible pairings, giving a $2^{\theta(|\beta(t)| \log |\beta(t)|)}$ bound on the number of different states.

With this observation, it is straightforward to design a dynamic programming algorithm that finds a Hamiltonian cycle in time $2^{\Oh(t \log t)} n$ given a tree decomposition of the input graph
of width $t$. This is exactly the naive approach.

\subsection{Rank based approach}

The rank-based approach is strongly based on the naive one, with main change being a pruning on the number of possible pairings.
\begin{theorem}[\cite{bodlaender2012solving}]\label{thm:rank}
For a fixed node $t$ and bucket $b$, given a family $\mathcal{E}$ of pairings, one can find a subfamily $\mathcal{E}^\ast \subseteq \mathcal{E}$
of size at most $2^{|b^{-1}(1)|-1}$ with the following property: for every Hamiltonian cycle $H$ in $G$, if $\mathcal{P}$ is its intersection with $G_t$
and $(b,E)$ is the state of $\mathcal{P}$, and $E \in \mathcal{E}$, then there exists $E^\ast \in \mathcal{E}^\ast$ such that for every partial solution
$\mathcal{P}^\ast$ with state $(b,E^\ast)$, the graph $(H-E(\mathcal{P})) \cup E(\mathcal{P}^\ast)$ is a Hamiltonian cycle as well.

Furthermore, given $b$ and $\mathcal{E}$, one can assign to every $E \in \mathcal{E}$ a $2^{|b^{-1}(1)|-1}$-length 0-1 vector $v_E$
such that the family $\mathcal{E}^\ast$ is defined as the indices of any maximal independent (over $\mathbb{F}_2$) subfamily of $\{v_E | E \in \mathcal{E}\}$.
\end{theorem}
In other words, for a fixed bucket $b$, it is sufficient to keep only $2^{|b^{-1}(1)|-1}$ pairings, and pruning unnecessary
pairings can be done via Gaussian elimination on a matrix with $|\mathcal{E}|$ rows and $2^{|b^{-1}(1)|-1}$ columns over the field $\mathbb{F}_2$ (the two-element field modulo $2$).

In~\cite{cygan2013fast}, Theorem~\ref{thm:rank} is improved with a different construction of vectors $v_E$ that are of length $2^{|b^{-1}(1)|/2-1}$.
Furthermore,~\cite{cygan2013fast} showed how to use the special structure of the vectors $v_E$ to avoid Gaussian elimination at \textbf{introduce/forget vertex/edge nodes}, 
yielding $(2+\sqrt{2})^p p^{\Oh(1)} n$-time algorithms for graphs with a given \emph{path} decomposition of width $p$ (i.e., without any \textbf{join nodes}).

We implement the rank-based approach both with the vector construction of~\cite{bodlaender2012solving} and the improved one
of~\cite{cygan2013fast}. Both implementations use Gaussian elimination, as it is not known how to avoid it at join nodes.

All implementations perform the same computations specific to the node type, which are straightforward in all cases.
At join nodes, the algorithm first sorts the partial solutions by buckets and then tries to match the partial solutions
only for buckets that fit each other (e.g., do not exceed the bound of $2$ on a degree of a vertex).

After successfully computing the set of partial solutions for a current node the algorithm runs a reduce function.
In the naive approach, it only deletes the duplicates by sorting set of partial solutions and checking
if the two consecutive are same or not. In rank-based approach it divides all partial solutions into buckets
(same as during processing the \textbf{join node}). For each bucket it computes the necessary matrix and performs
Gaussian elimination on it to get a representative set of partial solutions.

\subsubsection{Keeping track of partial solutions vs self-reducibility}\label{ss:self}

In the implementation, the core of the naive and rank-based approaches is the same. 
We use two variants of the implementations:
either \textbf{keep track of partial solutions} (so that the entire Hamiltonian cycle can be returned in the end)
  or, in order to save space, just remember a flag and a Hamiltonian cycle is found via \textbf{self-reducibility}.

To limit the effect of self-reducibility in case of the decision-only variant, we employ a problem-specific strategy.
That is, we discover the Hamiltonian cycle edge-by-edge. For a path $P$ in $G$ with at least two edges, we can discover if $G$ contains a Hamiltonian cycle containing $P$
by deleting from $G$ all edges of $E(G)\setminus E(P)$ that are incident to internal vertices of $P$, and run the decision algorithm on the obtained subgraph.
Given a path $P$, we extend it one by one by doing a binary search over the next edge incident to an endpoint of $P$.
This gives $\Oh(n \log \Delta)$ calls to the decision version of the problem for graphs with $n$ vertices and maximum degree $\Delta$.

\subsection{Cut\&Count approach}\label{ss:cut-and-count}

The main idea of the Cut\&Count approach~\cite{cygan2011solving} is to replace the search for a Hamiltonian cycle
with counting the following objects: a cycle cover of the graph (i.e., a subset of edges where every vertex is of degree exactly two)
with an assignment of every cycle to either left or right. In this manner, a fixed cycle cover with $c$ cycles is counted
$2^c$ times; if we additionally force one fixed vertex to be always on the left side, we get $2^{c-1}$ instead.
That is, every Hamiltonian cycle is counted once, and every other cycle cover is counted an even number of times.

In~\cite{cygan2011solving}, the Isolation Lemma~\cite{IsolationLemma} is employed to essentially reduce to the case when we solve instances with a unique Hamiltonian cycle. 
Then, the parity of the count described above indicates whether the graph contains a Hamiltonian cycle.
However, this approach introduces a large polynomial overhead in the running time bound:
first, because of the need for self-reducibility to discover the cycle (which we handle as in the previous section) 
  and, second, because of the use of Isolation Lemma that adds an additional ``weight'' dimension to the dynamic programming memoization tables.

For the second overhead, as discussed~\cite{cygan2011solving}, it can be remedied by, instead of using the Isolation Lemma,
pick a field $\mathbb{F}$ of characteristic $2$ (i.e., a field of size $2^p$ for some integer $p$), associate with each edge $e \in E(G)$ a variable $x_e$, associate with each cycle cover a monomial
being a product of the variables associated with the edges used in the cycle cover, and compute the sum of the monomials over all cycle covers and all left/right assignment,
using a random assignment of values from $\mathbb{F}$ to variables $x_e$.
Then, if $\mathbb{F}$ is large enough (larger than the maximum degree of the monomial, which is $n$), the Schwarz-Zippel lemma ensures that with good probability the result is nonzero if and only if the graph has a Hamiltonian cycle
(i.e., there is a small probability of a false negative).

In our implementation, we follow this path, using a field of size $2^{64}$. This size is large enough so that the failure probability is negligible.
On the other hand, there exists an efficient implementation of operations on this field using
the PCLMULQDQ processor instruction for multiplication. 
Our implementation of the field operations follow~\cite{bjorklund2014engineering}.

Furthermore, as discussed in the introduction,
  the choice of computations over $GF(2^{64})$ rather than arguably simpler counting algorithms via the Isolation Lemma resulted also in 
technical problems in handling \textbf{join nodes}. 
As observed in~\cite{cygan2011solving}, a natural and direct approach to a \textbf{join node} with bag of size $t$ runs in time $9^t t^{\Oh(1)}$.
In~\cite{cygan2011solving}, this is speeded up by an involved fast convolution approach, reducing the $9^t$ factor to $4^t$.
At heart of this approach lies an algorithm to quickly compute the following convolution.

Let $f,g : \mathbb{Z}_4^m \to R$ for some ring $R$ and integer $m$. 
We define $f \ast g : \mathbb{Z}_4^m \to R$ as
$$(f \ast g)(x) = \sum_{y \in \mathbb{Z}_4^m} f(y) g(x-y).$$
Here, the addition in $\mathbb{Z}_4^m$ is done coordinate-wise. 
\cite{cygan2011solving} developed a FFT-like approach to computing the above convolution, yielding the following.
\begin{lemma}[\cite{cygan2011solving}]\label{lem:conv}
Given $f,g : \mathbb{Z}_4^m \to \mathbb{Z}$, the convolution $f \ast g$ can be computed
in $4^m m^{\Oh(1)}$ operations on $\mathbb{Z}$ on values of the order of $2^{\Oh(m)}$ times larger than the maximum absolute value of the input functions.
\end{lemma}
However, the proof of the above lemma involves division by a factor of $4^m$, making it inapplicable directly to $R = GF(2^{64})$.
To circumvent this obstacle, we developed a new variant of Lemma~\ref{lem:conv}, building on the internal structure of the field
$GF(2^{64})$. Recall that a field $GF(2^p)$ can be defined as the ring $\mathbb{Z}[x]$ divided by the ideal generated by $2$
and an irreducible (in $\mathbb{F}_2[x]$) polynomial $Q$ of degree $p$.
\begin{lemma}\label{lem:conv2}
Let $p \geq 1$ and assume that the elements of field $GF(2^p)$ are given as polynomials from $\mathbb{F}_2[x]$ of degree less than $p$,
and multiplication in $GF(2^p)$ is done modulo a known polynomial $Q$ of degree $p$.
Given two function $f,g : \mathbb{Z}_4^m \to GF(2^p)$, the convolution $f \ast g$ can be computed
in time $4^m (pm)^{\Oh(1)}$.
\end{lemma}
\begin{proof}
We follow the same algorithm as in the proof of Lemma~\ref{lem:conv} from~\cite{cygan2011solving}, but treating the values
of $f$ and $g$ as elements of $\mathbb{Z}[x]$, not $GF(2^{64})$. 
This allows the necessary division steps in the algorithm, and an inspection of the proof of~\cite{cygan2011solving}
shows that the algorithm operates on $\Oh(m)$-bit integers and polynomials of degree $\Oh(p)$.
Then, at the very end, we reduce every resulting polynomial modulo $2$ and modulo $Q$ to obtain elements of $GF(2^{64})$.
\end{proof}
However, in the above we need to depart from the efficient implementation of operations in $GF(2^{64})$, and explicitly
operate on polynomials in $\mathbb{Z}[x]$ of larger degree. While theoretically sound, this is expected to give a large
overhead in experiments. 
Consequently, we test two variants of the Cut\&Count algorithm: the one using a naive approach to the join nodes
in time $9^t t^{\Oh(1)}$ and the one using Lemma~\ref{lem:conv2}.

To conclude the proof of Theorem~\ref{thm:cc}, we observe that to ensure
correctness with constant probability, the Cut\&Count algorithm
of~\cite{cygan2011solving} requires field $GF(2^p)$ with $p = \Omega(\log n)$.

\section{Setup}\label{sec:setup}

\subsection{Hardware and code}

All of the computations of our implementations were performed on a PC with an Intel Core i5-6500 processor and
16 GB of random-access memory. The operating system used during the experiments was Arch Linux.
All implementations have been done in C++, the code is available at~\cite{fnp-webpage,our-repo}.

The baseline \texttt{hamiltionian\_cycle()} method of the SageMath package
has been run on a PC with an Intel Core i7-6700 processor and 32 GB of random-access memory, running Ubuntu 18.04. 
We have used version 8.8 of SageMath with Gurobi 8.1.1 as LP solver backend.

\subsection{Data sets}

To evaluate our algorithms we decided to use the well known set of 
\textsc{Hamiltonian Cycle} instances from
Flinders Hamiltonian Cycle
Project~\cite{Haythorpe18} consisting of 1001 instances.
To find tree decompositions of small width, we first applied our implementation of the
minimum fill-in heuristic (cf.~\cite{GaspersGJMR16}).
The heuristic returned tree decompositions of width at most $8$ for
623 instances, and indicated that 30 more instances may have treewidth within ranges
allowing usage of our \textsc{Hamiltonian Cycle} algorithms.

We took the aforementioned 623 instances as our main benchmark.
For sake of optimizing hyperparameters of our algorithms (i.e., the frequency of applying the Gaussian elimination step in the rank-based approaches), we sampled a~subset of 30 elements.

To the aforementioned 30 instances with larger but potentially
tractable treewidth, we applied
the heuristic of Ben Strasser~\cite{Strasser} that won the second place in 2017 PACE Challenge~\cite{dellpace}. This resulted in another 19 instances with tree decompositions of width between
17 and 29.
Out of these instances, 15 turned out to be tractable by our algorithms.

Furthermore, we also sampled 7 random instances in the following way: starting from a Hamiltonian cycle $C$, we added a 
number of random edges with endpoints arranged on the cycle $C$ so that the treewidth is bounded. These instances are meant to generate
many partial solutions at separator, and are expected to give large advantage to rank-based approaches.
More precisely, the 7 instances in set $E$ below are generated in the following way:
\begin{enumerate}
\item Pick two integers $a$ and $b$ with $a$ equal $2$ modulo $4$ and a probability threshold $0 < p < 1$.
\item Create $ab$ vertices $\{v_{i,j}~|~1 \leq i \leq a, 1 \leq j \leq b\}$.
\item For every $1 \leq i \leq a$, connect vertices $v_{i,j}$ for $1 \leq j \leq b$ into a path (with indices $j$ in the natural order).
\item For every $1 \leq i \leq a/2$, connect $v_{i, 1}$ with $v_{i + a/2, 1}$. For every even $2 \leq i \leq a$, connect 
$v_{i-1, b}$ with $v_{i, b}$. Note that at this point all vertices are connected in a single cycle due to the assumption
that $a$ equals $2$ modulo $4$. This cycle is a Hamiltonian cycle of the constructed graph. 
\item For every $1 \leq i < i' \leq a/2$ and $1 \leq j \leq b$, add edge $v_{i,j} v_{i',j}$ with probability $p$.
\end{enumerate}
Note that the above procedure generates graph with treewidth bounded by $a$
(for every $1 \leq r \leq ab - a$, create a bag $X_r = \{v_{s\mathrm{\ mod\ } a, \lceil s/a \rceil}~|~r \leq s \leq r + a\}$
 and make a tree decomposition of the constructed graph being a path with $ab-a$ vertices with bags $X_r$).

To sum up, we operate on five data sets, all but the last being subsets of
the Flinders Hamiltonian Cycle Project~\cite{Haythorpe18}:
\begin{description}
\item[set $A$] is the whole set of graphs with small treewidth recognized by our heuristic
  (623 instances, treewidth at most 8),
\item[set $B$] is a subsample of $A$ (30 instances, treewidth at most 8),
\item[set $C$] is the set of larger treewidth graphs with decompositions
found by \cite{Strasser} (19 instances, treewidth between 17 and 29).
\item[set $D$] is the subset of the set $C$ that turned out to be 
tractable by our implementations (15 instances, treewidth between 17 and 29).
\item[set $E$] is a set of $7$ random graphs sampled as described above.
\end{description}
All instances from~\cite{Haythorpe18} are available through their webpage.
At~\cite{fnp-webpage} we provide a list of the used instances in each set, the set $E$, 
and the used tree decompositions for sets $C$ and $D$.
See Tables~\ref{tb:stats} and~\ref{tb:statsCE} for basic statistics of the tests used.

\begin{table}[bth]
\begin{center}
\begin{tabular}{l||c|c|c|c||c|c|c|c||c|c|c|c||}
\multirow{2}{*}{set} & \multicolumn{4}{c||}{$|V(G)|$} & \multicolumn{4}{c||}{$|E(G)|$} & \multicolumn{4}{c||}{tree decomposition width} \\
& min & avg & med & max
& min & avg & med & max
& min & avg & med & max \\\hline
A & 66 & 2406.98 & 2224 & 8886 & 99 & 5246.36 & 3871 & 35018 & 4 & 7.21 & 8 & 8 \\
B & 286 & 3059.77 & 2988 & 6620 & 430 & 7022.33 & 4646 & 28718 & 5 & 7.67 & 8 & 8 \\
C & 462 & 2173.05 & 1578 & 9528 & 756 & 3380.16 & 2688 & 13968 & 17 & 23.53 & 25 & 29 \\
E & 360 & 534.29 & 600 & 700 & 566 & 922.14 & 886 & 1397 & 15 & 17.29 & 17 & 20 \\
\hline\hline
\multirow{2}{*}{set} & \multicolumn{4}{c||}{minimum degree} & \multicolumn{4}{c||}{average degree} & \multicolumn{4}{c||}{maximum degree} \\
& min & avg & med & max
& min & avg & med & max
& min & avg & med & max \\\hline
A & 2 & 2.83 & 3 & 4 
  & 3 & 3.95 & 3 & 8.73 
  & 3 & 159.61 & 4 & 1908 \\
B & 2 & 2.8 & 3 & 4 
  & 3 & 4.14 & 3 & 8.71 
  & 3 & 263.13 & 4 & 1488 \\
C & 2 & 2 & 2 & 2 
  & 2.93 & 3.20 & 3.17 & 3.47 
  & 3 & 69.16 & 8 & 192 \\
E & 2 & 2 & 2 & 2 
  & 2.95 & 3.61 & 3.64 & 4.54 
  & 6 & 7.43 & 8 & 8 \\
\hline\hline
\multirow{2}{*}{set} & \multicolumn{4}{c||}{girth} & \multicolumn{4}{c||}{diameter} & \multicolumn{4}{c||}{} \\
& min & avg & med & max
& min & avg & med & max
&  &  &  &  \\\hline
A 
& 3 & 3.29 & 3 & 5 
& 6 & 194.13 & 89 & 1113
&  &  &  & \\
B
& 3 & 3.33 & 3 & 5
& 6 & 235.83 & 159 & 828
&  &  &  &  \\
C 
& 3 & 3.95 & 4 & 4
& 10 & 24.26 & 18 & 92
&  &  &  &  \\
E 
& 3 & 3 & 3 & 3 
& 33 & 47.57 & 53 & 54 
&  &  &  &  \\
\hline\end{tabular}
\caption{Statistics of test sets. avg stands for average, med stands for median.}\label{tb:stats}
\end{center}
\end{table}

\subsection{Breakdown of tests depending on their source}\label{ss:breakdown}

According to the authors of the Flinders Hamiltonian Cycle Project dataset~\cite{Haythorpe18},
the instances in the dataset come from multiple sources. 
In our repository~\cite{our-repo}, we provide a breakdown of tests into the following categories:
\begin{description}
\item[\texttt{generalized\_petersen}] 
Generalized Petersen graphs with $n = 3 \mathrm{\ or\ } 5\mathrm{\ mod\ }6$,
            so precisely three Hamiltonian cycles exist~\cite{inst-gp} (202 tests).
\item[\texttt{flower\_snarks}] Flower snarks modified by the addition of a single edge to introduce approximately $2^{n/8}/3$ Hamiltonian cycles~\cite{inst-flower} (150 tests).
\item[\texttt{uniquely\_fleischner}] Uniquely Hamiltonian graphs, construction by Fleischner~\cite{inst-f} (50 tests).
\item[\texttt{uniquely\_at}] Uniquely Hamiltonian graphs, construction by Aldred and Thomassen~\cite{inst-at} (50 tests).
\item[\texttt{sheehan}] Maximally uniquely Hamiltonian graphs, construction by Sheehan~\cite{inst-sheehan} (11 tests).
\item[\texttt{combined\_fleischner}] Combined smaller graphs, mostly from \texttt{uniquely\_fleischner} category (371 tests). 
\item[\texttt{reduction\_domset}] Graphs obtained from a reduction from the Dominating Set problem (15 tests).
\item[\texttt{reduction\_nqueens}] Graphs obtrained from a reduction from the N-queens problem (25 tests).
\item[\texttt{reduction\_insanity}] Graphs obtrained from a reduction from from the Generalized Instant Insanity problem (60 tests).
\item[\texttt{reduction\_bellringing}] Graphs obtrained from a reduction from the Bellringing problem (8 tests).
\item[\texttt{reduction\_unium}] Graphs obtrained from a reduction from commercial videogame Unium~\cite{inst-unium} (45 tests).
\item[\texttt{other}] A few unclassified tests (14 tests).
\end{description}

\afterpage{%
\clearpage%
\begin{landscape}%
\begin{table}[bth]
\begin{center}
\begin{tabular}{l|l|c|c|c|c|c|c|c|c}
category &
test & $|V(G)|$ & $|E(G)|$ & min deg & avg deg & max deg & girth & diameter & tw \\\hline

\rotcat{1}{\texttt{reduction\_bellringing}} &
0074 & 462 & 756 & 2 & 3.27 & 5 & 3 & 13 & 28 \\
\rotcat{8}{\texttt{reduction\_insanity}} &
0109 & 606 & 933 & 2 & 3.08 & 7 & 4 & 17 & 17 \\
& 0110 & 606 & 925 & 2 & 3.05 & 7 & 4 & 18 & 17 \\
& 0144 & 804 & 1256 & 2 & 3.12 & 7 & 4 & 18 & 21 \\
& 0145 & 804 & 1252 & 2 & 3.11 & 8 & 4 & 18 & 21 \\
& 0172 & 1002 & 1575 & 2 & 3.14 & 9 & 4 & 19 & 25 \\
& 0173 & 1002 & 1579 & 2 & 3.15 & 8 & 4 & 18 & 25 \\
& 0199 & 1200 & 1902 & 2 & 3.17 & 8 & 4 & 18 & 29 \\
& 0200 & 1200 & 1902 & 2 & 3.17 & 8 & 4 & 18 & 26 \\
\rotcat{7}{\texttt{reduction\_unium}} &
0253 & 1578 & 2688 & 2 & 3.41 & 163 & 4 & 10 & 29 \\
& 0268 & 1644 & 2767 & 2 & 3.37 & 192 & 4 & 10 & 25 \\
& 0271 & 1662 & 2770 & 2 & 3.33 & 183 & 4 & 10 & 29 \\
& 0272 & 1662 & 2863 & 2 & 3.45 & 183 & 4 & 10 & 25 \\
& 0290 & 1770 & 3020 & 2 & 3.41 & 176 & 4 & 10 & 25 \\
& 0298 & 1806 & 3071 & 2 & 3.40 & 182 & 4 & 10 & 23 \\
& 0340 & 2010 & 3488 & 2 & 3.47 & 159 & 4 & 10 & 26 \\
\rotcat{3}{\texttt{other}} &
0703 & 4024 & 5900 & 2 & 2.93 & 3 & 4 & 62 & 19 \\
& 0989 & 7918 & 11608 & 2 & 2.93 & 3 & 4 & 80 & 19 \\
& 1001 & 9528 & 13968 & 2 & 2.93 & 3 & 4 & 92 & 18 \\
\rotcat{7}{generated by us} & 
E0001 & 360 & 566 & 2 & 3.14 & 7 & 3 & 33 & 16 \\
& E0002 & 600 & 886 & 2 & 2.95 & 7 & 3 & 54 & 18 \\
& E0003 & 700 & 1139 & 2 & 3.25 & 8 & 3 & 54 & 20 \\
& E0004 & 300 & 681 & 2 & 4.54 & 6 & 3 & 52 & 20 \\
& E0005 & 700 & 1397 & 2 & 3.99 & 8 & 3 & 53 & 17 \\
& E0006 & 600 & 1120 & 2 & 3.73 & 8 & 3 & 53 & 15 \\
& E0007 & 360 & 655 & 2 & 3.64 & 8 & 3 & 34 & 15 \\
\hline\end{tabular}
\caption{Basic statistics for tests from sets C and E.
Min/avg/max deg stands for minimum, average, and maximum degree, respectively,
 and tw stands for the width of the used tree decomposition.}\label{tb:statsCE}
\end{center}\end{table}%
\end{landscape}%
\clearpage%
}

The partition of tests into the above categories turned out to be highly aligned
with the partition depending on the treewidth:
\begin{itemize}
\item The set $A$ (treewidth at most 8) consists of 171 (out of 371) \texttt{combined\_fleischner} instances, all 150 \texttt{flower\_snarks} instances, all 202 \texttt{generalized\_petersen} instances, 
  all 50 \texttt{uniquely\_at} instances, and all 50 \texttt{uniquely\_fleischner} instances.
\item The set $C$ (treewidth between 17 and 29) consists of 
1 (out of 8) \texttt{reduction\_bellringing} instance,
8 (out of 60) \texttt{reduction\_insanity} instances,
7 (out of 45) \texttt{reduction\_unium} instances,
and 3 (out of 14) \texttt{other} instances.
\item The instances in the set $C$ but not in the set $D$ (i.e., intractable for our approaches)
are the 3 \texttt{other} instances and one \texttt{reduction\_unium} instance.
\end{itemize}

\subsection{Fine-tuning the frequency of Gaussian elimination}

As discussed in the introduction, in the rank-based approach the theoretical
running time bound is worse than the one of Cut\&Count approach partially due to the
need of applying Gaussian elimination on the set of partial solutions. 
It is expected that the Gaussian elimination would also take substantial part of time
resources in experiments. 

In theory, the Gaussian elimination step is applied whenever the size of the set of partial
solutions exceeds theoretical guarantees. However, in practice it seems reasonable
that sometimes it pays off to apply this computationally expensive step less often;
that is, allow the set of partial solutions to grow significantly beyond the theoretical bounds, 
and once in a while trim it at bulk with a single Gaussian elimination step.
This intuition has been supported by the results of Fafianie et al~\cite{fafianie2015speeding}
for the case of~\textsc{Steiner Tree}.

Consequently, we start our experiments with fine-tuning the frequency of Gaussian elimination
in both rank-based approaches we study. Since the width of the tree decomposition
can play substantial role in deciding the optimal frequency, we do it separately 
for sets $B$ and sets $C$.

In the next experiments, we use the optimum found frequencies
for the algorithms based on both rank-based approaches.
Note that the frequencies may differ between the low-treewidth regime (sets $A$ and $B$)
and the medium-treewidth regime (set $C$).

\subsection{Comparison of the approaches}

Having found the optimal frequency of the Gaussian elimination in the rank-based approaches,
we run all four algorithms on every test in sets $A$, $B$, and $C$ and compare results.
In set $C$, every run has a timeout of 30 minutes. 
In set $A$, the timeout equals 10 minutes.

\section{Results}\label{sec:results}

In our experiments, it quickly became apparent that the variant of the naive and rank-based approach
that stores all partial solutions (i.e., no self-reducibility; see Section~\ref{ss:self}) is significantly faster for small treewidth (sets $A$ and $B$),
while the self-reducibility one is faster for larger treewidth (sets $C$, $D$, and $E$). 
Thus, in what follows, we used the first one for experiments on small treewidth graph and the latter
for larger treewidth graphs.

\subsection{Fine-tuning the frequency of Gaussian elimination}

\subsubsection{Small treewidth}

Recall that in sets $A$ and $B$, the maximum size of the
bag in the decomposition is $9$. Consequently, in every state
$(b, E)$ the size of $b^{-1}(1)$ is at most $8$ (as it must be even).
The treatment of the states with $|b^{-1}(1)| \in \{0,2\}$ does not
use any of the involved rank-based techniques. 
Thus, we decided to separately fine-tune the frequency of applying
the Gaussian elimination step to buckets with $|b^{-1}(1)|$ of size
$4$, $6$, and $8$ each. 
More formally, for $\ell \in \{4,6,8\}$
we fix a threshold $\tau$ and, for fixed bucket $b$ with $|b^{-1}(1)| = \ell$
apply the Gaussian elimination step to the states $(b,E)$ only if the number
of these states is at least $\tau$.
While experimenting with one $\ell$, the threshold for another sizes remains
fixed.
We perform tests on set $B$ and report the total time used to find
Hamiltonian cycles in all instances.
The results are presented in Table~\ref{tb:fine1}.

\begin{table}[bth]
\begin{center}
\begin{tabular}{l|l|l|l|ll}
\multirow{2}{*}{$\ell$} & \multirow{2}{*}{$2^{\ell-1}$} & \multirow{2}{*}{$2^{\ell/2-1}$} & \multirow{2}{*}{$\tau$} & \multicolumn{2}{c}{Total running time on set $B$ (SS.ms)} \\
&&&& rank-based $4^t$ & rank-based $(2+\sqrt{2})^t$ \\\hline
\multirow{2}{*}{4} & \multirow{2}{*}{8} & \multirow{2}{*}{2} & 3 & 1910.968 & 1318.385 \\
&&& 4 & 1827.949 & 1569.542 \\\hline
\multirow{7}{*}{6} & \multirow{7}{*}{32} & \multirow{7}{*}{4} & 5 & 1901.457 & 1264.803\\
&&& 7 & 1915.583 & 1286.522\\
&&& 9 & 1960.515 & 1298.849\\
&&& 11 & 1890.034 & 1316.813\\
&&& 13 & 1876.483 & 1339.439\\
&&& 15 & 1889.843 & 1401.620\\
&&& 17 & 1923.244 & 1425.338\\\hline
\multirow{5}{*}{8} & \multirow{5}{*}{128} & \multirow{5}{*}{8} & 9 & 1896.748 & 1269.761\\
&&& 18 & 1899.633 & 1290.696\\
&&& 36 & 1996.629 & 1274.545\\
&&& 72 & 1925.507 & 1268.261\\
&&& 144 & 1863.837 & 1283.934\\
\hline\end{tabular}
\caption{Fine-tuning results for test set $B$. Note that the second and third columns correspond to compression guarantees of the two studied algorithms, respectively.}\label{tb:fine1}
\end{center}
\end{table}

\onlylipics{\vspace{-10mm}}

From the results, it seems that lowering the frequency of Gaussian elimination does not help neither of the approaches, and evidently worsened the case for the improved base algorithm and $\ell=4,6$.
Lowering the frequency of Gaussian elimination only helped a bit in the case of $\ell=6$ and $\tau=13$ for the (first) rank-based algorithm.

We see a number of good explanations for that. First, we think that case $\ell=8$ appeared very rarely in the computations, and thus the impact of fine-tuning it has negligible effect in the overall result. 

For the remaining cases, note that the matrices passed to the Gaussian elimination have at most $32$ columns in the case of the first algorithm, and only $4$ columns in the second. Thus, the Gaussian elimination step is very cheap in this regime of $\ell$. Consequently, one does not gain much from lowering the frequency, while evidently losing by needing to maintain bigger memoization tables. This explains the worsening of the second algorithm for $\ell=4,6$ and increasing $\tau$.

However, one would expect that the first algorithm would also worsen with the increase of $\tau$, but this is not supported by data. To explain this behavior, note that the values of $\tau$ used here are lower than the theoretical guarantees of the algorithm. Consequently, the pruning of the memoization tables in the first algorithm seem to give very little in these cases. 

In other words, the pruning capabilities of the vectors $v_E$ used by the first algorithm are much weaker for low values of $\ell$ than the capabilities of the second algorithm. This is most striking in the case $\ell=4$: there are $3$
pairings of a $4$-element set; the first algorithm keeps all of them if present,
while the second one notices that one is redundant and deletes it.

To sum up, the data indicates that decreasing the frequency of the Gaussian elimination step does not help for small values of $\ell$, while the first algorithm
with the worse pruning capabilities does not offer much pruning in this
regime of values of $\ell$.

\subsubsection{Larger treewidth}

For fine-tuning in graphs of larger treewidth, we use set $D$. 
Here, we propose slightly different threshold behavior: we fix a parameter
$\alpha$ and, for fixed bucket $b$ with $\ell = |b^{-1}(1)|$,
  we apply the Gaussian elimination step
if the number of states $(b,E)$ exceed $\alpha \cdot 2^{\ell/2-1}$
(i.e., $\alpha$ is a multiplicative parameter relative to the pruning size
 guarantee of the second algorithm). 
The results are gathered in Table~\ref{tb:fine2}.

\begin{table}[bth]
\begin{center}
\begin{tabular}{l|ll||l|ll}
\multirow{2}{*}{$\alpha$} & \multicolumn{2}{c}{Total running time on set $D$ (SS.ms)} & \multirow{2}{*}{$\alpha$} & \multicolumn{2}{c}{Total running time on set $D$ (SS.ms)} \\
& rank-based $4^t$ & rank-based $(2+\sqrt{2})^t$ & 
& rank-based $4^t$ & rank-based $(2+\sqrt{2})^t$ \\\hline
0.5  & 7363.078 & 2021.516 & 32   & 1802.851 & 1778.647\\
1    & 2704.165 & 1801.278 & 64   & 1797.416 & 1807.470\\
2    & 1925.618 & 1768.000 & 128  & 1794.877 & 1801.913\\
4    & 1813.478 & 1779.293 & 256  & 1801.104 & 1822.113\\
8    & 1792.872 & 1788.217 & 512  & 1795.312 & 1818.508\\
16   & 1806.994 & 1783.919 & 1024 & 1800.863 & 1800.698\\
\hline\end{tabular}
\caption{Fine-tuning results for test set $D$.
  The Gaussian elimination step is applied to buckets $b$ with $\ell = |b^{-1}(1)|$
    and at least $\alpha \cdot 2^{\ell/2-1}$ states $(b,E)$.}\label{tb:fine2}
\end{center}
\end{table}

\onlylipics{\vspace{-10mm}}

The results indicate that a mild increase of the threshold (i.e., $\alpha = 2$)
  increases the speed of the second algorithm, while further increase
  of the threshold slowly worsens the bounds. 
  For the first algorithm, the sweet spot seems to be slightly later,
  and further increase of the threshold does not necessarily worsen the algorithm.

The gain from mild increase in the case of both algorithms can be explained
by the fact that for larger values of $\ell$, the Gaussian elimination step
starts to be costly. In the case of the first algorithm, we think that its
pruning capabilities are limited for the Hamiltonian cycle problem, and thus
further increase of the threshold does not change much. 

To sum up, both algorithms definitely slow down if the Gaussian elimination
step is done too frequently. The data showed optimum values $\alpha=8$
for the first algorithm and $\alpha=2$ for the second.

\subsection{Comparison}

As discussed in Section~\ref{ss:cut-and-count}, we have implemented
two variants of the Cut\&Count algorithm: the one that uses the
fast convolution at \textbf{join nodes} (Lemma~\ref{lem:conv2})
  and the one that does it more naively in time bounded by
  $9^t t^{\Oh(1)}$.

We found out that the one with the fast convolution behaves very slowly
even on small tests. This can be easily explained by the hidden complexity
of ring computations inside Lemma~\ref{lem:conv2}. Consequently, while
theoretically sound, we dropped it from further experiments
and considered only the Cut\&Count algorithm without the fast convolution.

For test set A, we have used a timeout of $10$ minutes per instance. 
A CSV file with full results can be found on the project website~\cite{fnp-webpage}.
Table~\ref{tb:cmp1} presents a summary.
The Cut\&Count algorithm did not finish in time for $124$ tests,
and thus we compare its running time on the other $499$ tests.
The SageMath baseline algorithm finished $188$ tests within the time limit (all tests finished by SageMath
were also finished by the Cut\&Count algorithm).
For sets C and E, full results are in Table~\ref{tb:cmp2} (for set E only naive and improved rank-based algorithms were executed).

\begin{table}[bth]
\begin{center}
\begin{tabular}{l|lllll}
& naive & rank-based $4^t$ & rank-based $(2+\sqrt{2})^t$ & Cut\&Count & Sage \\\hline
188 tests & 1591.500 & 1869.004 & 1869.172 & 6019.928 & 14338.490 \\
499 tests & 5993.633 & 7383.249 & 5919.392 & 46650.101 & - \\
all tests &  11532.153 & 13675.58 & 10278.827 & - & - \\
\hline\end{tabular}
\caption{Total running times for test set $A$ (timeout 10 minutes per instance). The Cut\&Count program did not finish within allotted time on $124$ instances,
  Sage finished within allotted time on only 188 instances, all of these instances were solved by the Cut\&Count program. 
  All other programs solved all test cases. 
The first row shows the total running time on 188 instances solved by all programs.
The second row shows the total running time on 499 instances solved by the Cut\&Count program.
The last row shows the total running time on all instances.}\label{tb:cmp1}
\end{center}
\end{table}

\onlylipics{\vspace{-10mm}}

\afterpage{%
\clearpage%
\begin{landscape}%
\begin{table}[bth]
\begin{center}
\begin{tabular}{l|l|lll|llll|l}
\multirow{2}{*}{category} & 
\multirow{2}{*}{test} & \multirow{2}{*}{$|V(G)|$} & \multirow{2}{*}{$|E(G)|$} & \multirow{2}{*}{tw} &
\multirow{2}{*}{naive} & rank-based & rank-based & Cut \& & \multirow{2}{*}{Sage}\\
& & & & & & \qquad$4^t$ & $(2+\sqrt{2})^t$ & Count & \\\hline
\rotcat{1}{\texttt{reduction\_bellringing}} &
0074 & 462 & 756 & 28 & 38.737 & 109.655 & 110.040 & - & 1.32 \\
\rotcat{8}{\texttt{reduction\_insanity}} &
0109 & 606 & 933 & 17 & .063 & .086 & .085 & .611 & 1.07 \\
& 0110 & 606 & 925 & 17 & .066 & .089 & .090 & .471 & 1.15\\
& 0144 & 804 & 1256 & 21 & .190 & .253 & .231 & 205.128 & 1.64 \\
& 0145 & 804 & 1252 & 21 & .137 & .187 & .186 & 3.549 & 1.36 \\
& 0172 & 1002 & 1575 & 25 & 1.156 & 1.298 & .554 & - & 1.75\\
& 0173 & 1002 & 1579 & 25 & .459 & .598 & .475 & 215.115 & 1.62 \\
& 0199 & 1200 & 1902 & 29 & 13.513 & 15.419 & 3.369 & - & 3.86 \\
& 0200 & 1200 & 1902 & 26 & 3.673 & 6.900 & 1.544 & - & 2.36 \\
\rotcat{6}{\texttt{reduction\_unium}} &
0253 & 1578 & 2688 & 29 & 93.343 & 167.458 & 167.440 & - & 6.19\\
& 0268 & 1644 & 2767 & 25 & 36.449 & 70.157 & 69.111 & - & 1.95\\
& 0271 & 1662 & 2770 & 29 & 28.149 & 33.145 & 33.208 & - & 3.15\\
& 0272 & 1662 & 2863 & 25 & 554.271 & 1260.329 & 1230.722 & - & 12.26\\
& 0290 & 1770 & 3020 & 25 & 57.901 & 83.781 & 82.386 & - & 89.84\\
& 0298 & 1806 & 3071 & 23 & 10.035 & 18.611 & 18.492 & - & 25.61\\
\rotcat{4}{generated by us} & 
E0001 & 360 & 566 & & 371.775 & - & 64.390 & - & .95\\
& E0002 & 600 & 886 & & 204.197 & - & 28.882 & - & 1.00\\
& E0003 & 700 & 1139 & & - & - & 711.778 & - & 1.50\\
& E0007 & 360 & 655 & & 1575.475 & - & 328.191 & - & 1.17\\
\hline\end{tabular}
\caption{Running times for test sets $C$ and $E$. Hyphen means timeout (30 minutes).
  For the set $C$, the ``tw'' column indicates the width of the used tree decomposition
    (found by the algorithm of Strasser~\cite{Strasser}). 
  Tests where all our implementations were time-outed are not presented here.}\label{tb:cmp2}
\end{center}
\end{table}%
\end{landscape}%
\clearpage%
}

The first corollary from the results is that the Cut\&Count approach does not turn out to be practical, and is 
heavily outperformed by other approaches. 
We see some good explanations for that.
First, all other approaches are ``positive-driven'': they keep only values in their memoization tables 
that correspond to found partial solutions, and in many cases there can be much fewer such partial solutions that the worse-case theoretical bound.
In particular, these approaches can implicitly use some hidden structure of the input graph, such as planarity. 
The Cut\&Count approach, on the other hand, relies on computing coefficients for partial cycle covers, and --- even with our positive-driven implementation
that keeps only nonzero elements --- keeps track of much more partial solutions than the other approaches.
This effect is even stronger if one tries to use fast convolution at \textbf{join nodes}: the convolution fills up the entire table of $4^t$ values
being polynomials, even if the input functions were sparse.

Second, the Cut\&Count approach solves only a decision version of the problem, yielding large overhead from some self-reducibility application,
  while all other algorithms return the Hamiltonian cycle in question straight away.

For the other approaches, it is noticeable that the first rank-based approach (with $4^t$ guarantee on the size of the memoization tables) is clearly
outperformed by the naive approach. That is, the cost of the Gaussian elimination step does not pay back in savings of the size of memoization tables.
This can be explained as already discussed in the previous section: the vectors used in this algorithm are too weak to effectively prune the memoization tables,
which is particularly visible on buckets $b$ with small $\ell = |b^{-1}(1)|$.

Results from small treewidth graphs (set $A$) show also that the improved rank-based approach outperforms the naive one by roughly $10\%$.
For larger treewidth (set $C$), the situation is more complicated: on some tests the rank-based approach outperforms the naive one
by significant factor (0172, 0199, 0200), while sometimes it is opposite (0074, 0272). 
The generated random instances (set $E$) gave big advantage to the rank-based approach, as expected as that was the main purpose in their design.

A natural question is why we see only $10\%$ increase despite significant asymptotic gain in the analysis ($2^{\Oh(t \log t)}$ vs $(2+\sqrt{2})^t)$. 
Apart from the obvious answers to this questions (the values of $t$ we are studying are low for asymptotic analysis), we would like to point out another,
problem-specific reason. The difference between the naive approach and the rank-based one is only within handling states for one fixed bucket $b$,
and there are up to $3^t$ different buckets. Iterating over all non-empty buckets is a common part of both approaches, and can be responsible for most
of their running time.

Comparing with the baseline SageMath algorithm, the treewidth-based approaches are clearly superior on small treewidth instances
(set $A$). 
This should be expected: the treewidth-based algorithms (except for Cut\&Count) have linear dependency on the graph size
in the running time bound while the exponential dependency on treewidth gives still moderately small constant for the values
of treewidth in the set $A$. On the other hand, for larger graphs in the set $A$, the SageMath method runs in exponential time in the graph
size in the worst case.

The advantage of treewidth-based algorithms disappears on the set $C$ where graphs have larger treewidth.
Here, 
the SageMath method usually outperforms our implementations or, in the other cases, is only mildly slower.
However, recall that the SageMath method 
has been run on a different, slightly stronger machine than our implementations, so we refrain from a more detailed comparison
of our implementations with the baseline SageMath method on the set $C$.

Figure~\ref{plot1} shows breakdown of running times on set $A$ (with 600s on top meaning time limit exceeded)
with regards to various graph parameters discussed in the previous section: number of vertices, number of edges, diameter, treewidth,
 average, and maximum degree. The breakdown with regards to minimum degree and girth
 turned out to carry very little information and is omitted. 
Clear (and expected) increase in the running time with the increase of the graph size (number of vertices, number of edges) is visible.
Note that the algorithms with linear-time dependency on the graph size (naive and both rank-based approaches) have this dependency visible in the plots.
There is a similar dependency on diameter, but most likely it is just the same correlation (i.e., with the graph size), as larger graphs tend to have larger diameter. 
The same comment applies to the average degree plot.
The plot with treewidth on the x-axis shows clearly that the running time of the treewidth-based algorithm explodes exponentially with the treewidth of the graph.
Finally, there does not seem to be any clear message given by the breakdown by maximum degree.

Figure~\ref{plot2} shows the running time on set $A$ (again, with 600s on top meaning time limit exceeded) as a function of the number of vertices, split into test categories as discussed
in Section~\ref{ss:breakdown}.
The plots clearly indicate that \texttt{uniquely\_fleischner} and \texttt{uniquely\_at} instances
are very simple for the treewidth-based approaches.
The instances from categories \texttt{flower\_snarks} and \texttt{generalized\_petersen} 
are also solved efficiently by the naive and rank-based approaches. 
Finally, \texttt{combined\_fleischner} instances can be quite challenging even for the best of
our implementations.
Note that \texttt{combined\_fleischner} is a very wide category, with diverse instances formed via combining subinstances from multiple sources.

Since the set $A$ includes all instances from categories
\texttt{uniquely\_fleischner}, \texttt{uniquely\_at}, 
\texttt{flower\_snarks}, and \texttt{generalized\_petersen}
and these instances have been solved quickly even by the naive approach,
we conclude that instances from these categories should no longer be considered 
as ``difficult'' for the Hamiltonian cycle problem.

To sum up, the only approach competitive with the naive approach is the improved rank-based approach with the $(2+\sqrt{2})^t$ guarantee on the size
of memoization tables. However, its gain is limited, and there are multiple cases where the use of Gaussian elimination steps is not helpful at all.
The treewidth-based algorithms greatly outperform the generic solver from SageMath on small treewidth instances.

\afterpage{%
\clearpage%
\begin{landscape}
\begin{figure}[h]%
\centering%
\includegraphics[width=.3\linewidth]{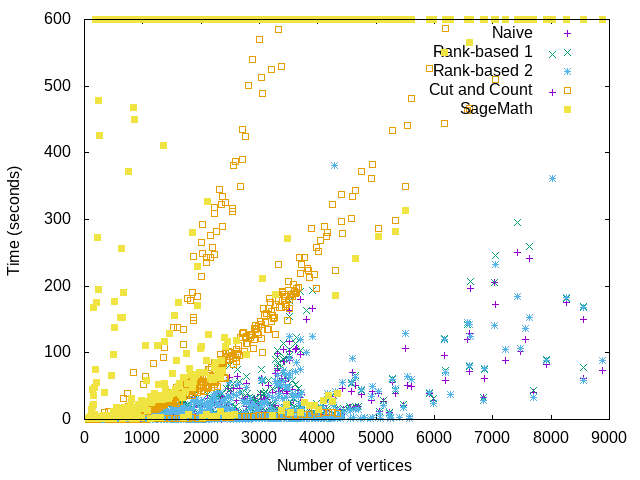}\quad%
\includegraphics[width=.3\linewidth]{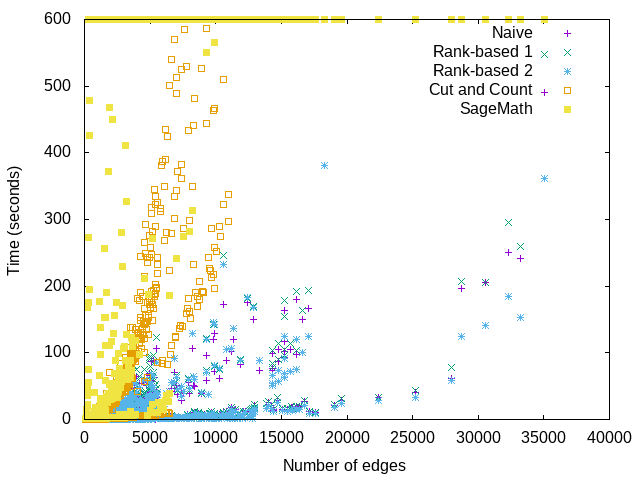}\quad%
\includegraphics[width=.3\linewidth]{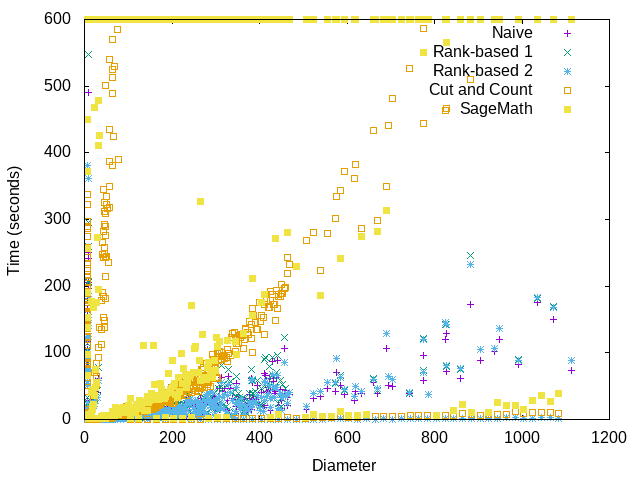}\\%
\includegraphics[width=.3\linewidth]{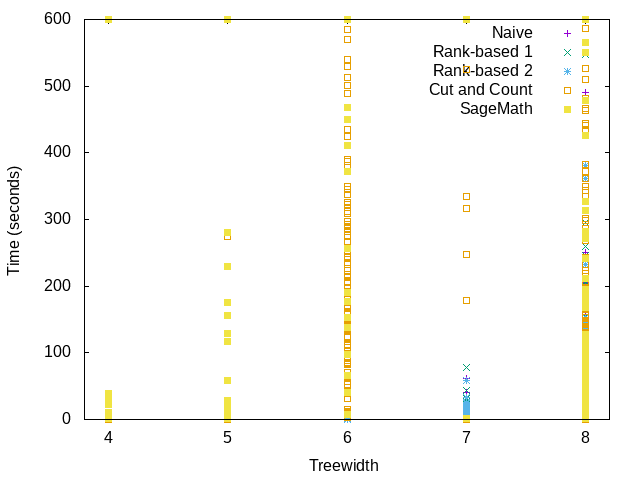}\quad%
\includegraphics[width=.3\linewidth]{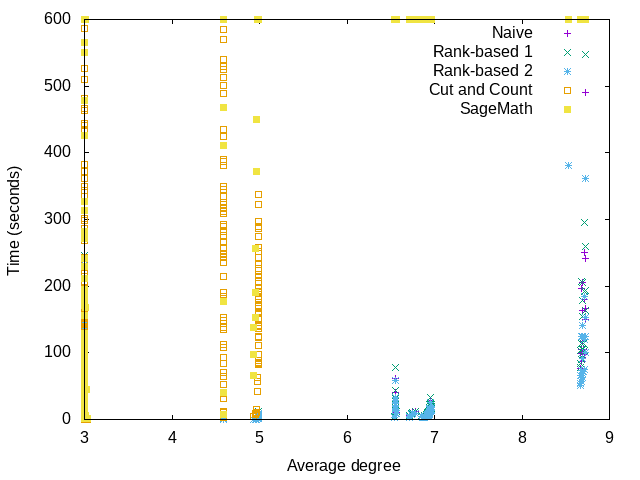}\quad%
\includegraphics[width=.3\linewidth]{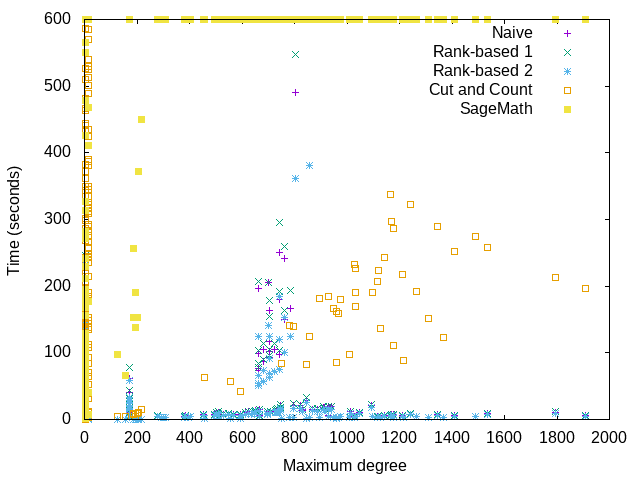}%
\caption{Plots of running times in set $A$ grouped by number of vertices, number of edges, diameter, treewidth, average degree, and maximum degree, respectively.
The corresponding plots for minimum degree and girth carry very little information and are not included.}\label{plot1}\vspace*{-.5cm}
\end{figure}%
\end{landscape}%
\clearpage%
}

\afterpage{%
\clearpage%
\begin{landscape}
\begin{figure}[h]%
\centering%
\includegraphics[width=.3\linewidth]{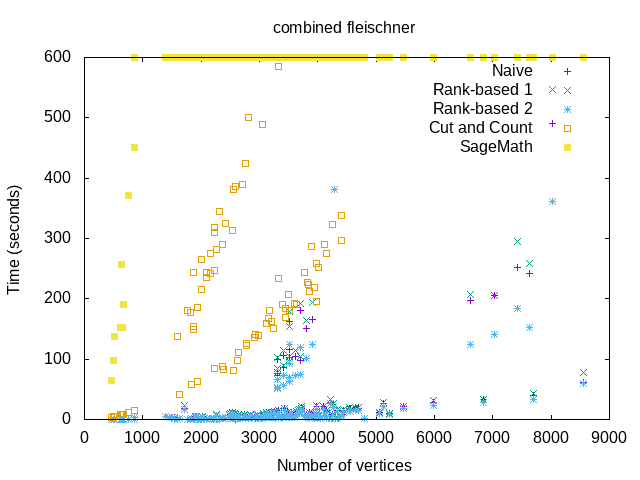}\quad%
\includegraphics[width=.3\linewidth]{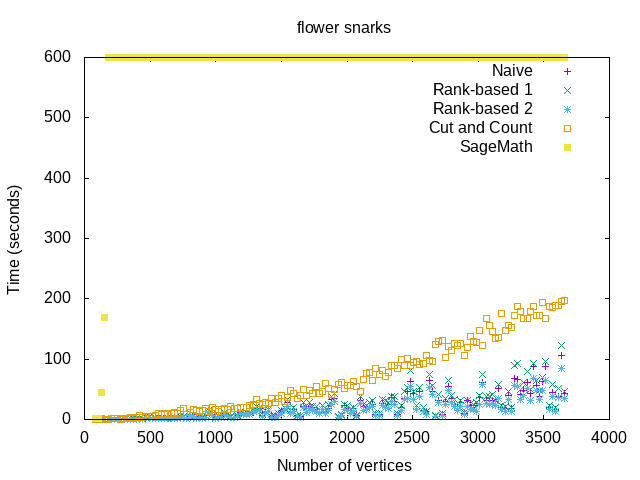}\quad%
\includegraphics[width=.3\linewidth]{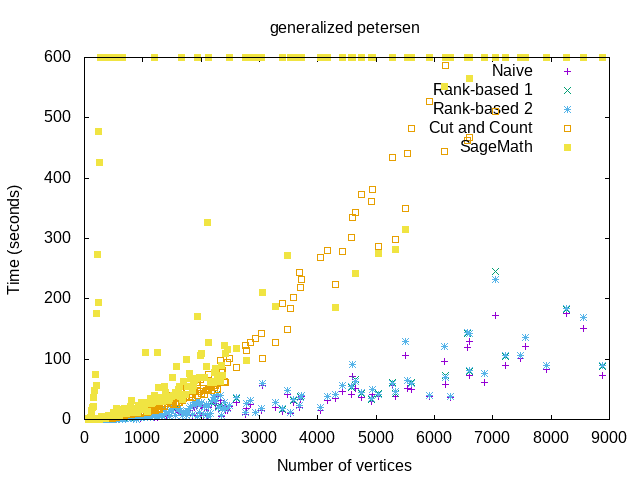}\\%
\includegraphics[width=.3\linewidth]{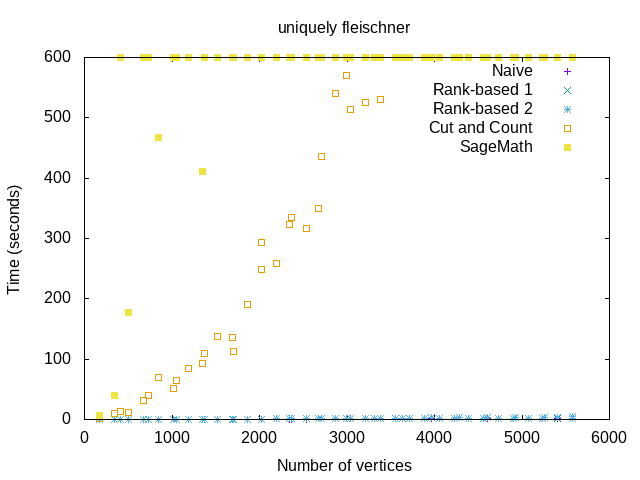}\quad%
\includegraphics[width=.3\linewidth]{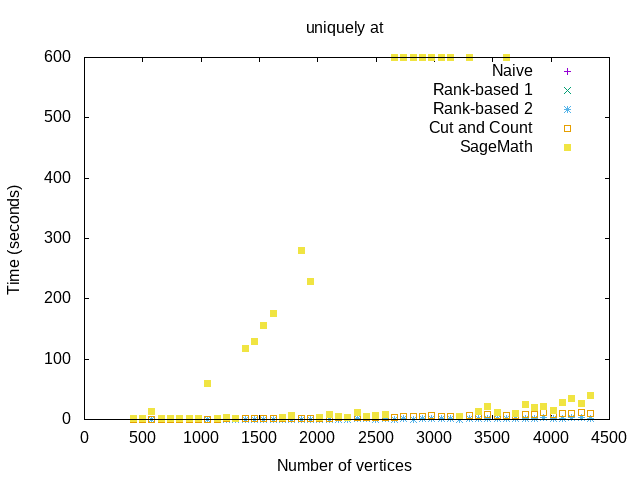}\\%
\caption{Plots of running times in subsets of set $A$, corresponding to \texttt{combined\_fleischner}, \texttt{flower\_snarks}, \texttt{generalized\_petersen}, \texttt{uniquely\_fleischner}, and \texttt{uniquely\_at} categories, respectively.}\label{plot2}\vspace*{-.5cm}
\end{figure}%
\end{landscape}%
\clearpage%
}

\section{Conclusions}\label{sec:conc}

We have experimentally evaluated multiple known approaches to solve \textsc{Hamiltonian Cycle} in graphs
of bounded treewidth. 
The results show that the Cut\&Count approach is impractical, while the improved rank-based approach of~\cite{cygan2013fast} consistently outperforms the
more generic one of~\cite{bodlaender2012solving}. Furthermore, the latter
seem to help little and is outperformed by the naive solution.

The comparison between the naive solution and the improved
rank-based one of~\cite{cygan2013fast} is more intricate. On graphs of small
treewidth, the second one outperforms the first one by $10\%$ margin.
For larger treewidth, the results are rather indecisive.

The results indicate potential in the improved rank-based algorithm
of~\cite{cygan2013fast} and point to the need of further theoretical
study of this approach. 
In~\cite{cygan2013fast}, the authors show how to perform pruning
without the need of Gaussian elimination at \textbf{introduce/forget nodes}.
The question of matching the $(2+\sqrt{2})^t t^{\Oh(1)}$ running time
bound for \textbf{join nodes} remains open, and a positive answer
to this question may lead to significantly faster implementation.
Also, we did not try to mix the Gaussian elimination steps
at \textbf{join nodes} with the other steps at \textbf{introduce/forget nodes}.

Finally, we found it quite remarkable that 638 out of 1001 instances
of Flinders Hamiltonian Cycle Challenge~\cite{Haythorpe18} (i.e., our sets $A$
    and $D$) could be solved with 
the naive bounded treewidth routine on a personal computer.
In particular, according to the Flinders Hamiltonian Cycle Challenge results~\cite{Haythorpe18},
such a routine would score a second place in the competition.
Furthermore, we were surprised that over $60\%$ (623 tests from our set $A$) have one-digit treewidth.

\paragraph{Acknowledgements.} This work is a full version of an extended abstract that appeared at SEA 2018~\cite{ZiobroP18}.

We are very grateful to an anynomous reviewer at ACM JEA whose remarks greatly helped us improving the manuscript.
We also thank Michael Haythorpe for assistance in preparing the breakdown of tests
into categories in Section~\ref{ss:breakdown}.

\bibliographystyle{abbrv}
\bibliography{A}

\begin{thebibliography}{10}

\bibitem{sage}
Sagemath package.

\bibitem{fnp-webpage}
Recent trends in kernelization: theory and experimental evaluation --- project
  website.
\newblock 2018.
\newblock \texttt{http://kernelization-experiments.mimuw.edu.pl}.

\bibitem{bellman1958combinatorial}
R.~Bellman.
\newblock Combinatorial processes and dynamic programming.
\newblock Technical report, RAND CORP SANTA MONICA CA, 1958.

\bibitem{bjorklund2010determinant}
A.~Bj{\"{o}}rklund.
\newblock Determinant sums for undirected hamiltonicity.
\newblock {\em {SIAM} J. Comput.}, 43(1):280--299, 2014.

\bibitem{bjorklund2010trimmed}
A.~Bj{\"o}rklund, T.~Husfeldt, P.~Kaski, and M.~Koivisto.
\newblock Trimmed moebius inversion and graphs of bounded degree.
\newblock {\em Theory of Computing Systems}, 47(3):637--654, 2010.

\bibitem{bjorklund2014engineering}
A.~Bj{\"o}rklund, P.~Kaski, L.~Kowalik, and J.~Lauri.
\newblock Engineering motif search for large graphs.
\newblock In {\em 2015 Proceedings of the Seventeenth Workshop on Algorithm
  Engineering and Experiments (ALENEX)}, pages 104--118. SIAM, 2014.

\bibitem{bodlaender2012solving}
H.~L. Bodlaender, M.~Cygan, S.~Kratsch, and J.~Nederlof.
\newblock Deterministic single exponential time algorithms for connectivity
  problems parameterized by treewidth.
\newblock {\em Inf. Comput.}, 243:86--111, 2015.

\bibitem{broersma2009fast}
H.~Broersma, F.~V. Fomin, P.~van't Hof, and D.~Paulusma.
\newblock Fast exact algorithms for hamiltonicity in claw-free graphs.
\newblock In {\em International Workshop on Graph-Theoretic Concepts in
  Computer Science}, pages 44--53. Springer, 2009.

\bibitem{cygan2015parameterized}
M.~Cygan, F.~V. Fomin, {\L}.~Kowalik, D.~Lokshtanov, D.~Marx, M.~Pilipczuk,
  M.~Pilipczuk, and S.~Saurabh.
\newblock {\em Parameterized {A}lgorithms}.
\newblock Springer, 2015.

\bibitem{cygan2013fast}
M.~Cygan, S.~Kratsch, and J.~Nederlof.
\newblock Fast hamiltonicity checking via bases of perfect matchings.
\newblock {\em J. {ACM}}, 65(3):12:1--12:46, 2018.

\bibitem{cygan2011solving}
M.~Cygan, J.~Nederlof, M.~Pilipczuk, M.~Pilipczuk, J.~M. van Rooij, and J.~O.
  Wojtaszczyk.
\newblock Solving connectivity problems parameterized by treewidth in single
  exponential time.
\newblock In {\em Foundations of Computer Science (FOCS), 2011 IEEE 52nd Annual
  Symposium on}, pages 150--159. IEEE, 2011.

\bibitem{dell2017first}
H.~Dell, T.~Husfeldt, B.~M. Jansen, P.~Kaski, C.~Komusiewicz, and F.~A.
  Rosamond.
\newblock The first parameterized algorithms and computational experiments
  challenge.
\newblock In {\em LIPIcs-Leibniz International Proceedings in Informatics},
  volume~63. Schloss Dagstuhl-Leibniz-Zentrum fuer Informatik, 2017.

\bibitem{dellpace}
H.~Dell, C.~Komusiewicz, N.~Talmon, and M.~Weller.
\newblock The pace 2017 parameterized algorithms and computational experiments
  challenge: The second iteration.

\bibitem{fafianie2015speeding}
S.~Fafianie, H.~L. Bodlaender, and J.~Nederlof.
\newblock Speeding up dynamic programming with representative sets: an
  experimental evaluation of algorithms for steiner tree on tree
  decompositions.
\newblock {\em Algorithmica}, 71(3):636--660, 2015.

\bibitem{inst-f}
H.~Fleischner.
\newblock Uniquely hamiltonian graphs of minimum degree 4.
\newblock {\em Journal of Graph Theory}, 75(2):167--177, 2014.

\bibitem{GaspersGJMR16}
S.~Gaspers, J.~Gudmundsson, M.~Jones, J.~Mestre, and S.~R{\"{u}}mmele.
\newblock Turbocharging treewidth heuristics.
\newblock {\em Algorithmica}, 81(2):439--475, 2019.

\bibitem{Haythorpe18}
M.~Haythorpe.
\newblock {FHCP} challenge set: The first set of structurally difficult
  instances of the {H}amiltonian cycle problem.
\newblock {\em Bulletin of the {ICA}}, 83:98--107, 2018.
\newblock Preprint at \url{http://arxiv.org/abs/1902.10352}, challenge set at
  \url{http://fhcp.edu.au/fhcpcs}.

\bibitem{held1962dynamic}
M.~Held and R.~M. Karp.
\newblock A dynamic programming approach to sequencing problems.
\newblock {\em Journal of the Society for Industrial and Applied Mathematics},
  10(1):196--210, 1962.

\bibitem{inst-at}
D.~A. Holton and R.~E.~L. Aldred.
\newblock Planar graphs, regular graphs, bipartite graphs and hamiltonicity.
\newblock {\em Australasian J. Combinatorics}, 20:111--132, 1999.

\bibitem{inst-flower}
R.~Isaacs.
\newblock Infinite families of nontrivial trivalent graphs which are not tait
  colorable.
\newblock {\em Amer. Math. Monthly}, 82:221--239, 1975.

\bibitem{LMS}
D.~Lokshtanov, D.~Marx, and S.~Saurabh.
\newblock Known algorithms on graphs of bounded treewidth are probably optimal.
\newblock {\em {ACM} Trans. Algorithms}, 14(2):13:1--13:30, 2018.

\bibitem{LMS2}
D.~Lokshtanov, D.~Marx, and S.~Saurabh.
\newblock Slightly superexponential parameterized problems.
\newblock {\em {SIAM} J. Comput.}, 47(3):675--702, 2018.

\bibitem{IsolationLemma}
K.~Mulmuley, U.~V. Vazirani, and V.~V. Vazirani.
\newblock Matching is as easy as matrix inversion.
\newblock {\em Combinatorica}, 7(1):105--113, 1987.

\bibitem{robertson1984graph}
N.~Robertson and P.~D. Seymour.
\newblock Graph minors. {III}. {P}lanar tree-width.
\newblock {\em Journal of Combinatorial Theory, Series B}, 36(1):49--64, 1984.

\bibitem{inst-sheehan}
J.~Sheehan.
\newblock Graphs with exactly one hamiltonian circuit.
\newblock {\em Journal of Graph Theory}, 1(1):37--43, 1977.

\bibitem{inst-unium}
J.~Statz.
\newblock Unium, 2015.
\newblock Accessed September 23rd 2019.

\bibitem{Strasser}
B.~Strasser.
\newblock Computing tree decompositions with flowcutter: {PACE} 2017
  submission.
\newblock {\em CoRR}, abs/1709.08949, 2017.

\bibitem{rooij}
J.~M.~M. van Rooij, H.~L. Bodlaender, and P.~Rossmanith.
\newblock Dynamic programming on tree decompositions using generalised fast
  subset convolution.
\newblock In A.~Fiat and P.~Sanders, editors, {\em Algorithms - {ESA} 2009,
  17th Annual European Symposium, Copenhagen, Denmark, September 7-9, 2009.
  Proceedings}, volume 5757 of {\em Lecture Notes in Computer Science}, pages
  566--577. Springer, 2009.

\bibitem{inst-gp}
M.~Watkins.
\newblock A theorem on tait colorings with an application to the generalized
  {P}etersen graphs.
\newblock {\em J. Combin. Th.}, 6:152–164, 1969.

\bibitem{ZiobroP18}
M.~Ziobro and M.~Pilipczuk.
\newblock Finding hamiltonian cycle in graphs of bounded treewidth:
  Experimental evaluation.
\newblock In G.~D'Angelo, editor, {\em 17th International Symposium on
  Experimental Algorithms, {SEA} 2018, June 27-29, 2018, L'Aquila, Italy},
  volume 103 of {\em LIPIcs}, pages 29:1--29:14. Schloss Dagstuhl -
  Leibniz-Zentrum fuer Informatik, 2018.

\bibitem{our-repo}
M.~Ziobro and M.~Pilipczuk.
\newblock Finding {H}amiltonian {C}ycle in graphs of bounded treewidth:
  Experimental evaluation. code repository, 2018.
\newblock \texttt{https://github.com/stalowyjez/hc\_tw\_experiments}.

\end{thebibliography}

\end{document}